\documentclass[lettersize,journal]{IEEEtran}
\usepackage{amsmath,amsfonts}
\usepackage{adjustbox}
\usepackage{tikz}
\usepackage{color}
\usepackage{algorithmic}
\usepackage{array}
\usepackage{physics}
\usepackage[caption=false,font=normalsize,labelfont=sf,textfont=sf]{subfig}
\usepackage{textcomp}
\usepackage{stfloats}
\usepackage{url}
\usepackage{verbatim}
\usepackage{graphicx}
\newtheorem{remark}{Remark}
\newtheorem{theorem}{Theorem}

\newtheorem{proof}{Proof}
\newtheorem{prot}{Protocol}      
\newcommand{\BPR}{\begin{prot}}   \newcommand{\EPR}{\end{prot}}
\newtheorem{definition}{Definition}

\usepackage[bookmarksopen,bookmarksnumbered]{hyperref}

\def\BibTeX{{\rm B\kern-.05em{\sc i\kern-.025em b}\kern-.08em
    T\kern-.1667em\lower.7ex\hbox{E}\kern-.125emX}}
\usepackage{balance}
\usepackage{cite}

\begin{document}
\title{{ Blockchain-Envisioned Post-Quantum Secure Sanitizable Signature for Audit Logs Management
		
}}

\author{ Vikas Srivastava,  Paresh Baidya, Sihem Mesnager, Debasish Roy, Sumit Kumar Debnath

	\thanks{Vikas Srivastava is with the ``Department of Mathematics, National Institute of Technology Jamshedpur, Jamshedpur 831 014, India (e-mail: vikas.math123@gmail.com).'' }
			
	\thanks{ Paresh Baidya is with the ``Department of Mathematics, National Institute of Technology Jamshedpur, Jamshedpur 831 014, India.'' }
	
	\thanks{ Debasish Roy is with the ``CER, Department of Mathematics, Indian Institute of Technology Kharagpur, Kharagpur, 721 302, India (email: debasish.roy@maths.iitkgp.ac.in)'' }
	
	\thanks{Sumit Kumar Debnath is with the ``Department of Mathematics, National Institute of Technology, Jamshedpur 831 014, India (e-mail: sd.iitkgp@gmail.com, sdebnath.math@nitjsr.ac.in).'' }
	
	\thanks{Sihem Mesnager is with the ``Department of Mathematics, University of Paris VIII, F-93526 Saint-Denis; University Sorbonne Paris Cit\'e, LAGA, UMR 7539, CNRS, 93430 Villetaneuse,  and T\'el\'ecom Paris, 91120 Palaiseau, France; ''}
}

\maketitle

\begin{abstract}
Audit logs are one of the most important tools for transparently tracking system events and maintaining continuous oversight in corporate organizations and enterprise business systems. There are many cases where the audit logs contain sensitive data, or the audit logs are enormous. In these situations, dealing with a subset of the data is more practical than the entire data set. To provide a secure solution to handle these issues, a sanitizable signature scheme (SSS) is a viable cryptographic primitive. Herein, we first present the \textit{first} post-quantum secure multivariate-based SSS, namely {\sf Mul-SAN}. Our proposed design provides unforgeability, privacy, immutability, signer accountability, and sanitizer accountability under the assumption that the $MQ$ problem is NP-hard. {\sf Mul-SAN} is very efficient and only requires computing field multiplications and additions over a finite field for its implementation. {\sf Mul-SAN} presents itself as a practical method to partially delegate control of the authenticated data in avenues like the healthcare industry and government organizations. We also explore using Blockchain to provide a tamper-proof and robust audit log mechanism.
\end{abstract}

\section{Introduction}

Audit logs serve as a vital tool for transparently tracking system events. In addition, it serves as a method for keeping an oversight in corporate organizations and enterprise businesses\cite{ahmad2018towards}. These logs are essential for data auditing, ensuring compliance, and meeting regulatory requirements. They enable stakeholders to monitor system status, observe user actions, and ensure accountability for user roles and performance through secure documentation. Additionally, they are frequently relied upon to restore data to a previous state in cases of unauthorized alterations. However, the audit logs mechanism needs to be fixed, which we discuss below.

\noindent \subsubsection*{Motivation}
\begin{itemize}
	\item Involvement of a trusted third party (for example, refer \cite{liu2022id,han2022survey}), make audit logs vulnerable to a single point of failure. Therefore, there is a secure and efficient audit log mechanism that can be tampered with or manipulated. 
	\item Auditing logs may contain sensitive information. Moreover, audit logs are usually very huge in size. Therefore, in practical situations, it is more realistic for auditors to work with a portion of the audit logs rather than the complete set. Redacting/Sanitizing all of the data not included in the subset is the same as taking a subset of the data. Redaction can, therefore, be used to establish a connection between the integrity of the master data and the integrity of a subset. To facilitate privacy-aware audit log management, a flexible redaction method that can guarantee the integrity of the redacted data will be a helpful tool.

	\item Existing state-of-the-art audit integrity cryptographic protocols\cite{liu2022id,su2020publicly,hahn2020enabling,zhang2023cl,zhang2023owl,wang2019lightweight,balasubramanian2019cloud,garg2020efficient,shen2019data} are based on number-theoretic hardness assumption. Once a large-scale quantum computer is available, these designs will be broken. Hence, there is a need for alternative cryptographic designs for secure audit log mechanisms that are resistant to attacks by quantum computers. Post-Quantum Cryptography (PQC) \cite{ding2017current,srivastava2023overview,bernstein2017post,micciancio2009lattice} is one such direction. Therefore, there is a need to design an appropriate PQC protocol for maintaining audit integrity.
	
	\end{itemize}

Conventional digital signature schemes are cryptographic primitives traditionally used to provide end-to-end message integrity and authenticity. However, digital signatures prohibit any modification in the original document. If the contents of the original document are altered, the signature becomes invalid. Nevertheless, there are a lot of real-life situations where the alteration of a signed document by an external party is needed. For instance, the court magistrate may want to hand over the right to summon someone to his office staff. A sanitizable signature scheme (SSS) can be employed in these situations. An SSS enables the signer to partially hand over the signing right to some trusted external party, called the sanitizer (censor). The idea and design of SSS were conceptualized by Atinese et al. \cite{ateniese2005sanitizable}. The general schematic idea behind an SSS is summarized in Figure \ref{blockdiag}.
Over the last decade, many SSS \cite{bossuat2021unlinkable,bultel2019efficient,lai2016efficient,fleischhacker2018efficient,schroder2019efficient,samelin2020policy} have emerged, but the security of almost all of the currently used SSS relies on number theoretic challenging problems. There has yet to be any construction of multivariate-based SSS in the current state of the art. There has yet to be construction of post-quantum secure SSS. This indicates the requirement of designing a post-quantum secure and efficient SSS.

\begin{figure*}[h!]
	\centering
	\includegraphics[scale=0.3]{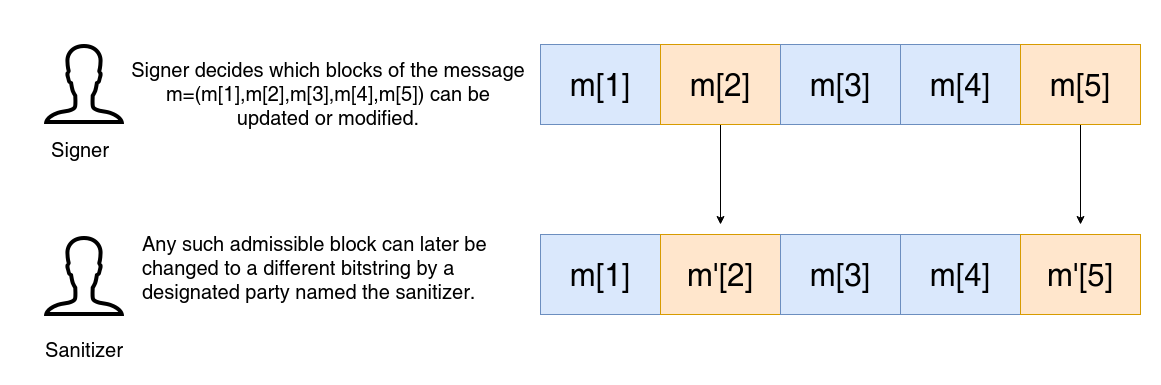}
	\caption{Sanitizable signature scheme. Fixed parts of the message are highlighted in blue.}
	\label{blockdiag}
\end{figure*}

\subsection{Our Contribution}
The contribution of this paper are summarised below
\begin{itemize}
\item In this paper, we put forth the \textit{first} post-quantum secure multivariate-based SSS. Our proposed construction {\sf Mul-SAN} consist of six algorithms: {\sf Mul-SAN}=({\sf Mul-SAN.KGen-Sign, Mul-SAN.KGen-Sanit, Mul-SAN.Signature, Mul-SAN.Sanitization, Mul-SAN.Verification, Mul-SAN.Judge}). On inputting a security parameter $\kappa$, the signer employs the algorithm {\sf Mul-SAN.KGen-Sign} to generate his key pair $({\sf sk_{sign}, pk_{sign}})$. Similarly, the sanitizer also produces his private key-public key pair $({\sf sk_{sanit}, pk_{sanit}})$ by running {\sf Mul-SAN.KGen-Sanit}. To allow modification in the original document/message {\sf msg} by a trusted censor, the signer decides {\sf AD}, a deterministic algorithm containing the description of the changeable message blocks and those fixed. In the following, signers first sign the fixed part of the document producing $\sigma_1$ and then the entire message producing $\sigma_2$. The signer finally outputs the message-signature pair {\sf (msg, $\sigma=(\sigma_1,\sigma_2, {\sf AD})$)}. A trusted sanitizer by employing the algorithm {\sf Mul-SAN.Sanitization} first make the modifications (using {\sf MODIFY}) producing the sanitized message {\sf msg'}, and then replaces the signature part $\sigma_2$ with a new sanitized signature generated under the ${\sf pk_{sanit}}$ (but leaving signature on the fixed part of the message unchanged).
	A verifier verifies the signature $\sigma$ for a given message ${\sf msg}$ by using the algorithm {\sf Mul-SAN.Verification}. In particular, the verifier first extracts the fixed part of {\sf msg}. In the following, the verifier inspects whether {\sf MODIFY} is admissible according to the {\sf AD}. Subsequently, the verifier checks the validity of the first part of signature ${\sigma}_1$ under the public key of the signer and then checks whether $\sigma_2$ is verified under the public key of the sanitizer or the signer; the algorithm {\sf Mul-SAN.Judge} can be employed later to derive the message-signature pair's origin.
	
	\item Our proposed design provides unforgeability, privacy, immutability, signer accountability and sanitizer accountability under the assumption that the $MQ$ problem is NP-hard. {\sf Mul-SAN} belongs to the family of MPKCs. Hence, it is naturally very efficient and only requires implementing field multiplications and additions over a finite field. Since {\sf Mul-SAN} 's security is based on the MQ-problem's intractability, our proposed design is the \textit{first} among SSS to offer strong resistance against quantum computers.
	
	\item We utilize the proposed Sanitizable signature and Blockchain technology to explore a possible solution for the audit log mechanism.
\end{itemize}

\section{Preliminaries \label{prelim}}

\subsection{Multivariate Signature Scheme}

A generic multivariate signature scheme comprises three algorithms: ({\sf Kg, Sig, Ver}).

\begin{description}
	\item {\sf ${\sf pub_k,sec_k \leftarrow Kg(\kappa}$)}: In the key generation phase, given a security parameter $\kappa$ a quadratic map $\mathcal{F}:\mathbb{F}_q^n\rightarrow \mathbb{F}_q^m$ with $n$ variables and $m$ equations is chosen in such a way that finding pre-image of some element under this map is easy to compute. Here $\mathbb{F}_q$ denotes the finite field of order $q$. Additionally two invertible affine transformations $\mathcal{S}_1 : \mathbb{F}_q^m\rightarrow \mathbb{F}_q^m$ and  $\mathcal{S}_2 : \mathbb{F}_q^n\rightarrow \mathbb{F}_q^n$ are picked to hide the structure of $\mathcal{F}$. The public key ${\sf pub_k}$ of the system is $\mathcal{S}_1\circ \mathcal{F}\circ\mathcal{S}_2$, while the secret key ${\sf sec_k}$ is given by the three tuple $\{\mathcal{S}_1, \mathcal{F},\mathcal{S}_2\}$ \vspace{0.2in}
	
	\item ${\sf x \leftarrow Sig(y, sec_k)}$: To sign a document $y\in \mathbb{F}_q^m$ under the secret key ${\sf sec_k}$, a signer proceeds in the following manner. The signer recursively calculates ${\sf x_0}=\mathcal{S}_1^{-1}({\sf y})\in \mathbb{F}_q^m,{\sf x_1}=\mathcal{F}_1^{-1}({\sf x_0})\in \mathbb{F}_q^n$ and ${\sf x}=\mathcal{S}_2^{-1}({\sf x_1})\in \mathbb{F}_q^n$. In the end, the signer outputs $(\sf x,y)$ as the signature-message pair.\vspace{0.2in}
	
	\item ${\sf 0/1 \leftarrow Ver(x, pub_k)}$: Given a signature ${\sf x}$ on a message ${\sf y}$, the verifier checks the validity of the equality ${\sf y}=\mathcal{P}({\sf x})$. If the equality holds, the signature is accepted. Otherwise, the verifier rejects it.
\end{description}

\subsection{Hardness Assumption \cite{garey1979computers}}

The security of our design depends on the hardness of the $MQ$ problem, which is formulated mathematically as:

\begin{definition} Given a system $\mathcal{R} = \left(r_{(1)}(\delta,\dots,\delta_n) ,\dots, r_{(k)}(\delta_1,\dots, \delta_n)\right)$ of $k$ quadratic equations {with each $r_{(i)}\in \mathbb{F}_q[\delta_1,\dots,\delta_n]$, find $(\bar{\delta}_1 ,\dots, \bar{\delta}_n) \in \mathbb{F}_q^n$} such that
	$$r_{(1)}\left(\bar{\delta}_1,\dots,\bar{\delta}_n \right) = \dots = r_{(k)}\left(\bar{\delta}_1,\dots,\bar{\delta}_n \right) = 0.$$
\end{definition}

\noindent We now describe the general construction of sanitizable signature scheme \cite{brzuska2009security}.
\subsection{Sanitizable Signature (SSS)} A generic  SSS is a collection of seven algorithms: SSS=({\sf KGen-Sign, KGen-Sanit, Signature, Sanitization, Verification, Judge}). The following definitions are taken from \cite{brzuska2009security}.

\begin{description}
	\item {\sf ${\sf (pk_{sign},sk_{sign})}\leftarrow$ KGen-Sign($1^\kappa$):} On input a security parameter $\kappa$, signer runs {\sf KGen-Sign} to produce its secret key ${\sf sk_{sign}}$ and the corresponding public key ${\sf pk_{sign}}$.\vspace{0.2in}
	\item {\sf ${\sf (pk_{sanit},sk_{sanit})}\leftarrow$ KGen-Sanit($1^\kappa$):} On input a security parameter $\kappa$, sanitizer runs {\sf KGen-Sanit} to generate its secret key ${\sf sk_{sanit}}$ and the corresponding public key ${\sf pk_{sanit}}$.\vspace{0.2in}
	
	\item $\sigma$ $\leftarrow${\sf Signature(${\sf msg, sk_{sign},pk_{sanit}, AD}$):} Given as input the secret key of the signer ${\sf  sk_{sign}}$ , public key of the sanitizer  ${\sf pk_{sanit}}$, and an admissible function ${\sf AD}$ , signer employs {\sf Signature } to output a signature $\sigma$ on a message {\sf msg} $\in \{0,1\}^*$. \vspace{0.2in}
	
	\item (${\sf msg'}, \sigma') \leftarrow$ {\sf Sanitization(${\sf msg,{\sf MODIFY},\sigma,pk_{sign},}\\{\sf sk_{sanit}}$):} Given a message ${\sf msg} \in \{0,1\}^*$ as input, a signature $\sigma$, the modification information ${\sf MODIFY}$, the public key ${\sf pk_{sign}}$ of the signer, and it's own secret key  ${\sf sk_{sanit}}$, a sanitizer first modifies the {\sf msg} according to the modification instruction and then generates the sanitized signature $\sigma'$ on the modified message ${\sf msg'}$ by putting in use {\sf Sanitization}. \vspace{0.2in}
	
	\item {\sf 0/1 $\leftarrow$ Verification(${\sf msg,\sigma, pk_{sign},pk_{sanit}}$):} On input  ${\sf pk_{sign}}$ and ${\sf pk_{sanit}}$, a verifier checks the validity of message-signature pair ({\sf msg}, $\sigma$) with respect to the public keys  ${\sf pk_{sign}}$ and ${\sf pk_{sanit}}$ by putting to use the algorithm {\sf Verification}. It produces {\sf true} as output if the signature is valid; otherwise, it rejects and outputs {\sf false}. \vspace{0.2in}

	\item {\sf org $\leftarrow$ Judge(${\sf msg,\sigma, pk_{sign},pk_{sanit}}$):} On input a message-signature pair ({\sf msg}, $\sigma$) and public keys of both parties, the algorithm {\sf Judge} outputs a decision {\sf org $\in$ \{{\sf sign, sanit}\}} indicating whether the signer or the sanitizer has created the message-signature pair.	\vspace{0.2in}
\end{description}

\begin{figure*}[ht]
	{	
		\label{ fig7}
		\begin{minipage}[b]{0.5\linewidth}
			
			\centering

			{\scriptsize

				${\sf EXP\mbox{-}Signer\mbox{-}Acc_{\mathcal{G}_{sign}}^{SSS}}$
				\begin{description}[noitemsep]
					\item[] ${\sf (pk_{sanit}, sk_{sanit})}\leftarrow {\sf KGen\mbox{-}Sanit(1^\kappa)}$
					\item[] $ ({\sf pk*_{sign}, {\sf msg^{*}}, \sigma^{*}})\leftarrow \mathcal{G}_{\sf sign}^{{\sf Sanitization(\cdot,\cdot, \cdot,sk_{sign})}}({\sf pk_{sanit}})$\newline
					letting $({\sf msg'}_{i}, {\sf \sigma'}_i, {\sf pk}_{{\sf sanit},i})$  $\forall i=1,2,\ldots, \Delta$ denote the queries and answers to and from oracle {\sf Sanitization}.
					\item[] return 1 if \newline (${\sf pk^*_{sign}}, {\sf msg^*})\neq ({\sf pk}_{{\sf sign},i}, {\sf msg'}_i)$ for all $i=1,2,\dots, q$  and \newline {\sf Verification(${\sf msg^*,\sigma^*,pk^*_{sign}, pk_{sanit}}$)=true} and \newline {\sf Judge(${\sf msg^*,\sigma^*, pk^*_{sig},pk_{sanit}}$)={\sf San}}
				\end{description}
				
			}
			\caption{Signer Accountability of SSS}
			\label{signeracc}
			\vspace{4ex}
	\end{minipage}}
	{
		\begin{minipage}[b]{0.5\linewidth}
			\centering
			{\scriptsize
				${\sf EXP\mbox{-}San\mbox{-}Acc_{\mathcal{G}_{sanit}}^{SSS}}$
				\begin{description}[noitemsep]
					\item[] ${\sf (pk_{sign}, sk_{sign})}\leftarrow {\sf KGen\mbox{-}Sign(1^\kappa)}$
					\item[] $ ({\sf pk*_{sanit}, {\sf msg^{*}}, \sigma^{*}})\leftarrow \mathcal{G}_{\sf sanit}^{{\sf Signature(\cdot,sk_{sign},\cdot,\cdot)}}({\sf pk_{sign}})$\newline
					letting $({\sf msg}_{i}, {\sf AD}_i, {\sf pk}_{{\sf sanit},i})$ and $\sigma_i$ for $i=1,2,\ldots, q$ denote the queries and answers to and from oracle {\sf Signature}.
					\item[] return 1 if \newline (${\sf pk_{sanit}^*}, {\sf msg^*})\neq ({\sf pk}_{{\sf sanit},i}, {\sf msg'}_j)$ for all $i=1,2,\dots, q$  and \newline {\sf Verification(${\sf msg^{*},\sigma^{*}, pk_{sign},pk*_{sanit}}$)=true} and \newline {\sf Judge(${\sf msg^{*},\sigma^{*}, pk_{sign},pk*_{sanit}}$)=true}
			\end{description}}
			\caption{Sanitizer Accountability of SSS}
			\label{sanacc}
			\vspace{4ex}
	\end{minipage}}
	{
		\begin{minipage}[b]{0.5\linewidth}
			
			\centering
			
			{\scriptsize
				${\sf EXP\mbox{-}Immutability_{\mathcal{G}}^{SSS}}$
				\begin{description}[noitemsep]
					\item[] ${\sf (pk_{sign}, sk_{sign})}\leftarrow {\sf KGen\mbox{-}Sign(1^\kappa)}$
					\item[] ${\sf (pk_{sanit}, sk_{sanit})}\leftarrow {\sf KGen\mbox{-}Sanit(1^\kappa)}$	
					\item[] $ ({\sf pk*_{sanit}, {\sf msg^{*}}, \sigma^{*}})\leftarrow \mathcal{G}^{{\sf Signature(\cdot,sk_{sign},\cdot,\cdot)}}({\sf pk_{sign}})$\newline
					letting $({\sf msg}_{i}, {\sf AD}_i, {\sf pk}_{{\sf sanit},i})$ and $\sigma_i$,  $i\in \{1,\ldots, \Delta\}$ denotes the queries to and responses from the oracle {\sf Signature}.
					\item[] return 1 if \newline {\sf Verification(${\sf msg^{*},\sigma^{*}, pk_{sign},pk*_{sanit}}$)=true} and \newline $\forall$ $i=1,\dots, \Delta$ we have \newline \hspace{0.5pt} ${\sf pk*_{sanit}\neq pk_{{sanit},i}}$ or \newline ${\sf msg^{*}\notin \{MODIFY(msg)\; |\; MODIFY \mbox{ with }}$ \\${\sf AD}_i({\sf MODIFY})=1\}$
			\end{description}}
			\caption{Immutability of SSS}
			\label{immut}
			\vspace{4ex}
	\end{minipage}}
	\vspace{2ex}
	{
		\begin{minipage}[b]{0.5\linewidth}
			
			\centering
			
			{\scriptsize				
				${\sf EXP\mbox{-}Unforgeability_{\mathcal{G}}^{SSS}}$
				
				\begin{description}[noitemsep]
					\item ${\sf (pk_{sign}, sk_{sign})}\leftarrow {\sf KGen\mbox{-}Sign(1^\kappa)}$
					\item ${\sf (pk_{sanit}, sk_{sanit})}\leftarrow {\sf KGen\mbox{-}Sanit(1^\kappa)}$	
					\item {\tiny $ ({\sf msg^*}, \sigma^*)\leftarrow \mathcal{G}^{{\sf Signature(\cdot,sk_{sign},\cdot,\cdot)},{\sf Sanitization(\cdot, \cdot, sk_{sanit},\cdot)}}({\sf pk_{sign}, pk_{sanit}})$}\newline
					letting $({\sf msg}_{i}, {\sf AD}_i, {\sf pk}_{{\sf sanit},i})$ and $\sigma_i$ with $i\in \{1,\ldots, \Delta\}$ denotes the queries to and responses from the oracle {\sf Signature} and $({\sf msg}_{j}, {\sf MODIFY}_j, \sigma_j,{\sf pk}_{{\sf sign},i})$ and $({\sf msg'_j},\sigma'_j)$ for $j\in \{\Delta+1,\ldots, r\}$ denote the queries to and responses from the oracle {\sf Sanitization}.
					\item return 1 if \newline {\sf Verification(${\sf msg^{*},\sigma^{*}, pk_{sign},pk*_{sanit}}$)=true} and \newline $\forall$ $i\in \{1,\ldots, q\}$ we have  (${\sf pk*_{sanit}}, {\sf msg^*})\neq ({\sf pk}_{{\sf sanit},i}, {\sf msg}_i)$ and \newline for all $j=q+1,\dots, r$ we have  (${\sf pk_{sign}}, {\sf msg^*})\neq ({\sf pk}_{{\sf sign},j}, {\sf msg'}_j)$
			\end{description}}
			\caption{Unforgeability of SSS}
			\label{unforge}
			\vspace{4ex}
	\end{minipage}}
\end{figure*}

We will now discuss the security properties of SSS. We will follow the discussion in \cite{brzuska2009security}.
\subsubsection{Unforgeability}
The security notion of unforgeability ensures that no adversary can produce a new valid signature on behalf of the signer or sanitizer.
\begin{definition}
	A SSS is unforgeable if given a PPT adversary $\mathcal{G}$, and the success probability of $\mathcal{ G}$ in the security game \ref{unforge} is negligible.
\end{definition}

\subsubsection{Immutability}  The security notion of immutability ensures that the sanitizer is only allowed to make changes in the admissible portions of the document.

\begin{definition}
	We say that an SSS is immutable if given any PPT adversary $\mathcal{G}$, the probability that the immutability experiment defined in Figure \ref{immut} returns $1$ is negligible.
\end{definition}

\subsubsection{Privacy} The security notion of privacy ensures that, aside from the signer and the sanitizer, no one can gain any information about the sanitized portions of the document.

\begin{definition}
	A SSS satisfies privacy if given a PPT adversary $\mathcal{G}$, the success probability of $\mathcal{ G}$ in the security game \ref{private} is negligible.
\end{definition}

\begin{figure*}
	\centering{\small
		${\sf EXP \mbox{-} Privacy_{\mathcal{G}}^{SSS}}$
		\begin{description}[]
			\item[] ${\sf (pk_{sign}, sk_{sign})}\leftarrow {\sf KGen\mbox{-}Sign(1^\kappa)}$
			\item[] ${\sf (pk_{sanit}, sk_{sanit})}\leftarrow {\sf KGen\mbox{-}Sanit(1^\kappa)}$	
			\item[] $b \leftarrow \{0,1\}$
			\item[] $a\leftarrow \mathcal{G}^{{\sf Signature(\cdot,sk_{sign},\cdot,\cdot)},{\sf Sanitization(\cdot, \cdot, sk_{sanit},\cdot)}, {\sf LoRSanit(\cdot,\cdot, \cdot,sk_{sign}, sk_{sanit},b)}}({\sf pk_{sign}, pk_{sanit}})$\newline
			where oracle ${\sf LoRSanit(\cdot,\cdot, \cdot,sk_{sign}, sk_{sanit},b)}$ on input $({\sf msg_{j,0 }}, {\sf Modify_{j,0}}), ({\sf msg_{j,1 }}, {\sf Modify_{j,0}},)$ and ${\sf AD}_j$ first computes $\sigma_{j, b}\leftarrow {\sf Signature({\sf msg_{j,b},sk_{sign}, pk_{sanit}, AD_j})}$ and then returns ${\sf (msg'_{j}, \sigma'_j)\leftarrow Sanitization(msg_{j,b}, MODIFY_{j,b}, \sigma_{j, b}, pk_{sign}, sk_{sanit})}$, and where $({\sf msg}_{j,0}, {\sf MODIFY}_{j,0}, {\sf AD}_j)\equiv({\sf msg}_{j,1}, {\sf MODIFY}_{j,1}, {\sf AD}_j)$
			\item[] return 1 if $a=b$
	\end{description}}
	\caption{Privacy of SSS}
	\label{private}
\end{figure*}
\subsubsection{ Accountability}

The security notion of accountability ensures that any dispute regarding the origin of the signature should be settled correctly. Accountability can be classified into two types: (i) sanitizer accountability and (ii) signer accountability. The former means that not even a corrupt sanitizer could make the judge wrongly accuse the signer. At the same time, the latter ensures that even a malicious signer is not allowed to accuse the sanitizer falsely.
\begin{definition}
	We say that a SSS is a signer (sanitizer) accountable, if given any PPT adversary $\mathcal{G}$, the probability that the immutability experiment defined in Figure \ref{signeracc}(Figure \ref{sanacc}) returns $1$ is negligible.
\end{definition}

\section{Proposed Multivariate-Based Sanitizable Signature Scheme{\sf Mul-SAN}) \label{proposed}}
{\bf A high-level overview:}  For the sake of brevity, we first introduce some notations. These notations are widely used in the literature \cite{brzuska2009security,bossuat2021unlinkable,brzuska2010unlinkability}. {\sf AD} and {\sf MODIFY} are descriptions of efficient algorithms such that

\begin{align*}
{\sf msg'}\leftarrow {\sf MODIFY(msg)}
\end{align*}
\[
{\sf AD(MODIFY) }=
\begin{cases}
1   & \mbox{ if modifications are valid \newline with respect  }\\ & \mbox{ to }  {\sf AD};\\
0 &   \text{ otherwise.}
\end{cases}
\]

In other words, the algorithm {\sf MODIFY} maps the original {\sf msg} to the modified one. ${\sf AD(MODIFY)}\in \{0,1\}$ indicates if the modification is admissible. Let ${\sf FIXED_{AD}}$ be an efficient procedure which is determined (uniquely) by ${\sf AD}$. ${\sf FIXED_{AD}}$ maps {\sf msg} to the fixed part of the message ${\sf msg_{FIX} = FIXED_{AD}(msg)}$. We assume that ${\sf FIXED_{AD}(msg')}\neq {\sf FIXED_{AD}(msg)}$ for ${\sf msg'} \in \{{\sf MODIFY(msg) }\; |\; {\sf MODIFY }\mbox{ with } {\sf AD(MODIFY)}=1\} $

We now describe the workflow of our proposed design in brief. We use a secure multivariate-based signature scheme as the fundamental building block of our design. Motivated by the work of Brzuska et al. \cite{brzuska2009security,brzuska2009sanitizable}, our construction {\sf Mul-SAN} consist of six algorithms:
{\sf Mul-SAN}=({\sf Mul-SAN.KGen-Sign, Mul-SAN.KGen-Sanit, Mul-SAN.Signature, Mul-SAN.Sanitization, Mul-SAN.Verification, Mul-SAN.Judge}). Our scheme does not need an explicit {\sf Proof} algorithm. On input, a security parameter $\kappa$, the signer (respective sanitizer) employs the algorithm {\sf Mul-SAN.KGen-Sign} ( respectively {\sf Mul-SAN.KGen-Sanit}) to generate the key pair $({\sf sk_{sign}, pk_{sign}})$ (respectively $({\sf sk_{sanit}, pk_{sanit}})$). The signer then chooses {\sf AD}, which contains the description of blocks of message {\sf msg }, which the sanitiser can modify.
In the next step, the signer first generates the signature $\sigma_1$ on the fixed part ${\sf msg_{fix}}$, and subsequently signs the entire message producing $\sigma_2$. The signer finally outputs the message-signature pair {\sf (msg, $\sigma=(\sigma_1,\sigma_2, {\sf AD})$)}. A trusted sanitizer uses the algorithm {\sf Mul-SAN.Sanitization} to produce the sanitized message {\sf msg'}. In the following, he replaces the signature part $\sigma_2$ with a new signature $\sigma_2^{'}$ generated on the modified message {\sf msg'}. During the sanitization process, the sanitizer leaves the signature $\sigma_1$ on the fixed part of the message unchanged. Given a message-signature pair $({\sf msg}, \sigma)$, a verifier employs the algorithm {\sf Mul-SAN.Verification} to first extract the fixed part of {\sf msg}. In the next step, the verifier checks whether
${\sf AD(MODIFY)}=1$. Following this, he first checks the validity of the first part of the signature ${\sigma}_1$ under the signer's public key. He later checks the validity of the second part of the signature $\sigma_2$ under the public keys of the sanitizer or the signer. On input a message-signature pair ({\sf msg}, $\sigma$) and public keys of both parties, the algorithm {\sf Mul-SAN.Judge} outputs a decision {\sf org $\in$ \{{\sf sign,sanit}\}} indicating whether the signer or the sanitizer has created the message-signature pair. We have inserted the bit $0$ and $1$ to indicate the difference between the fixed message and the entire message. We now present our design.

{\small\BPR{{\sf Mul-SAN}}{
		\begin{description}
			\item {\sf ${\sf (pk_{sign},sk_{sign})}\leftarrow$ Mul-SAN.KGen-Sign($\kappa$):} On input of the security parameter $\kappa$, the signer runs the key generation algorithm of the underlying secure multivariate-based signature scheme to generate a secret key  $\{\mathcal{S}, \mathcal{F}, \mathcal{T}\}$ and the corresponding public key $ \mathcal{P}=\mathcal{S}\circ \mathcal{F}\circ\mathcal{T}:\mathbb{F}_q^n\rightarrow \mathbb{F}_q^m$. Here $\mathcal{S}:\mathbb{F}_q^m\rightarrow \mathbb{F}_q^m$ and $\mathcal{T}:\mathbb{F}_q^n\rightarrow \mathbb{F}_q^n$ are two randomly chosen affine invertible transformation, and $\mathcal{F}:\mathbb{F}_q^n\rightarrow \mathbb{F}_q^m$ denotes the central map. It then sets $\{\mathcal{S}, \mathcal{F}, \mathcal{T}\}$ as his secret key ${\sf sk_{sign}}$, and $ \mathcal{P}$ as the corresponding public key ${\sf pk_{sign}}$.
			
			\item {\sf ${\sf (pk_{sanit},sk_{sanit})}\leftarrow$ Mul-SAN.KGen-Sanit($\kappa$):} On input of the security parameter $\kappa$, the sanitizer runs the key generation algorithm of the underlying secure multivariate-based signature scheme to generate a secret key $\{\mathcal{Q}, \mathcal{X}, \mathcal{Y}\}$ and the corresponding public key $\mathcal{R}=\mathcal{Q}\circ \mathcal{X}\circ\mathcal{Y}:\mathbb{F}_q^n\rightarrow \mathbb{F}_q^m$. Similar to the above algorithm, $\mathcal{Q}:\mathbb{F}_q^m\rightarrow \mathbb{F}_q^m$ and $\mathcal{Y}:\mathbb{F}_q^n\rightarrow \mathbb{F}_q^n$ are two randomly chosen affine invertible transformation, and $\mathcal{X}:\mathbb{F}_q^n\rightarrow \mathbb{F}_q^m$ denotes the central map. It then sets $\{\mathcal{Q}, \mathcal{X}, \mathcal{Y}\}$ as his secret key ${\sf sk_{sanit}}$, and $ \mathcal{R}$ as the corresponding public key ${\sf pk_{sanit}}$.
			
			\item $\sigma$ $\leftarrow${\sf Mul-SAN.Signature(${\sf msg, sk_{sign},pk_{sanit}, AD}$):} Given secret key of the signer ${\sf  sk_{sign}}$ , public key of the sanitizer  ${\sf pk_{sanit}}$, and an admissible function ${\sf AD}$ as input, signer outputs the signature $\sigma$ on a message {\sf msg} $\in \{0,1\}^*$ by executing the following steps:
			\begin{enumerate}
				\item Given the message {\sf msg}, the signer sets the fixed part of ${\sf msg}$ by employing the
				algorithm ${\sf Fixed_{AD}}$ determined by {\sf AD}. Let ${\sf msg_{fix}=Fixed_{AD}({\sf msg})}$ denote the fixed part of the original message ${\sf msg}$.
				\item Signer first signs the fixed part ${\sf msg_{fix}}$ by going through following steps:
				\begin{enumerate}
					\item Sets ${\sf msg_0=\mathcal{H}(0||{\sf msg_{fix}}||{AD}||{\sf pk_{sanit}})}$ where $\mathcal{H}$ is a  cryptographically secure hash function $\mathcal{H}:\{0,1\}^*\rightarrow \mathbb{F}^m$.
					\item  In the following, signer executes recursively $\alpha_0=\mathcal{S}^{-1}({\sf msg_0}), \beta_0= \mathcal{F}^{-1}(\alpha_0)$, and $\sigma_1=\mathcal{T}^{-1}(\beta_0)$.
					\item Sets ${\sf \sigma_1}$ as the signature on the fixed part of the message.
				\end{enumerate}
				\item In the following, signer signs the full message ${\sf msg}$ by executing the following steps:
				\begin{enumerate}
					\item Sets ${\sf msg_1=\mathcal{H}(1||msg||pk_{sanit}||pk_{sign})}$ where $\mathcal{H}$ is a  cryptographically secure hash function $\mathcal{H}:\{0,1\}^*\rightarrow \mathbb{F}^m$.
					\item  In the next step, signer executes recursively $\alpha_1=\mathcal{S}^{-1}({\sf msg_1}), \beta_1= \mathcal{F}^{-1}(\alpha_1)$ and $\sigma_2=\mathcal{T}^{-1}(\beta_1)$.
					\item Sets $\sigma_2$ as the signature on the full message.
				\end{enumerate}
				\item Outputs $\sigma=(\sigma_1, \sigma_2, {\sf AD})$ as the signature on a message {\sf msg}.
			\end{enumerate}

			\item (${\sf msg'}, \sigma') \leftarrow$ {\sf Mul-SAN.Sanitization(${\sf msg,{\sf MODIFY},\sigma,pk_{sign},sk_{sanit}}$):} Given a message ${\sf msg} \in \{0,1\}^*$, a signature $\sigma=(\sigma_1, \sigma_2, {\sf AD})$, the modification information ${\sf MODIFY}$, the public key ${\sf pk_{sign}}$ of the signer, and its own secret key ${\sf sk_{sanit}}$, sanitizer produces the sanitized signature $\sigma'$ by going through following steps:
			\begin{enumerate}
				\item Sanitizer first recovers the fixed part of the message ${\sf msg_{fix}=Fixed_{AD}({\sf msg})}$ . It then checks that {\sf MODIFY} is admissible according to the {\sf AD}. If the modification is admissible and matches ${\sf AD}$, he proceeds to the next step; otherwise, he aborts.
				\item Checks the validity of the signature $\sigma_1$ for message $(0, {\sf msg_{fix}, AD, pk_{sanit}})$ under the public key ${\sf pk_{sign}}$. If the signature is invalid, it stops and outputs $\perp$; otherwise, it proceeds to the next step.
				\item Uses the modification information to find the modified message ${\sf msg'\leftarrow MODIFY(msg)}$
				\item Generates the signature $\sigma'_2$ on the modified message {\sf msg'} in the following manner:
				\begin{enumerate}
					\item Sets ${\sf msg_2=\mathcal{H}(1||msg'||pk_{sanit}||pk_{sign})}$
					\item  In the following, signer executes recursively $\alpha_2=\mathcal{Q}^{-1}({\sf msg_2}), \beta_2= \mathcal{X}^{-1}(\alpha_2)$ and $\sigma'_2=\mathcal{Y}^{-1}(\beta_2)$.
				\end{enumerate}
				\item Outputs ${\sf msg'}$ together with $\sigma'=(\sigma_1,\sigma'_2,{\sf AD})$ as the sanitized message signature pair.
			\end{enumerate}

			\item {\sf 0/1 $\leftarrow$ Mul-SAN.Verification(${\sf msg,\sigma, pk_{sign},pk_{sanit}}$):} A verifier validates  a signature $\sigma=(\sigma_1, \sigma_2, {\sf AD})$ on the message ${\sf msg}$ w.r.t. ${\sf pk_{sign}}$ and ${\sf pk_{sanit}}$ in the following way:
			\begin{enumerate}
				\item Parses $\sigma = (\sigma_1 , \sigma_2 , {\sf AD} )$.
				\item Recovers the fixed part of the message ${\sf msg_{fix}=Fixed_{AD}({\sf msg})}$ . It then checks that {\sf MODIFY} is admissible according to the {\sf AD}.
				\item If $\mathcal{P}{\sf (\sigma_1)\stackrel{?}{=}msg_0}$ accepts $\sigma_1$ as the valid signature, otherwise abort.
				
				\item Checks that either of the equality  $\mathcal{P}(\sigma_2)\stackrel{?}{=}{\sf msg_1}$ or $ \mathcal{R}(\sigma_2)\stackrel{?}{=}{\sf msg_2}$ holds true. If {\sf true}, it publishes $1$, indicating the validity of the entire signature; otherwise, it rejects and outputs $0$.

			\end{enumerate}
			
			\item {\sf Mul-SAN.Judge(${\sf msg,\sigma, pk_{sign},pk_{sanit}}$):}  The judge on input $ {\sf msg,\sigma, pk_{sign},pk_{sanit}}$ decides the origin of the signature by executing following steps:
			
			\begin{enumerate}
				\item If $\mathcal{P}(\sigma_2)\stackrel{?}{=}{\sf msg}_1$ validates as {\sf true}, it outputs ${\sf Sig}$ indicating that signer generated the signature; else outputs ${\sf San}$ indicating that signature was generated by the sanitizer. We are assuming that the judge works on a valid message-signature pair.
			\end{enumerate}
		\end{description}
	}
	\EPR}

\noindent \begin{remark} The correctness of our design follows straightforwardly from the correctness of the underlying multivariate-based signature used as the building block.
\end{remark}

\section{Security Analysis\label{security}}
In this section, we will show that our proposed design {\sf Mul-SAN} is unforgeable, satisfies privacy and immutability, and offers both sanitizer accountability and signer accountability. We begin this section by proving that {\sf Mul-SAN} is immutably provided the underlying multivariate-based signature is unforgeable.
\begin{theorem}
	If the underlying MQ-based signature employed in the design and construction of {\sf Mul-SAN} is unforgeable, the proposed sanitizable signature scheme {\sf Mul-SAN} is immutable.
\end{theorem}
\begin{proof}We assume on the contrary that our proposed design {\sf Mul-SAN} is not immutable. Suppose that there exists a polynomial-time adversary $\mathcal{G}$ such that the advantage $\delta(\kappa)=
	{\sf Prob}[{\sf EXP\mbox{-}Immutability_{\mathcal{G}}^{SSS}}=1]$ is non-negligible. We show how to build an adversary $\mathcal{B}$, which breaks the unforgeability of the underlying secure MQ-based signature scheme. Note that $\mathcal{B}$ has oracle access to a signing oracle of the underlying secure MQ-based signature.
	
	\begin{description}
		\item {$\mathcal{ B}'s$ construction}: $\mathcal{ B}$ receives the public key ${\sf pk_{sign}}$ and runs  $ ({\sf pk^*_{sanit}, {\sf msg^*}, \sigma^*})\leftarrow \mathcal{G}^{}({\sf pk_{sign}})$. During the experiment, ${\cal B}$ simulates the  oracle ${\sf Signature(\cdot, sk_{sign}, \cdot, \cdot)}$ to $\mathcal{G}$ as follows:
		\item {\sf Signature$(\cdot, {\sf sk_{sign}}, \cdot, \cdot)$}: On the $i$th input $({\sf msg}_{i}, {\sf AD}_i, {\sf pk}_{{\sf sanit},i})$, $\mathcal{ B}$ first computes the fixed message part ${\sf msg_{fixed}} \leftarrow {\sf  FIXED}_{{\sf AD}_i}({\sf msg}_i)$. It then sends $({\sf msg}_i, {\sf AD}_i, {\sf pk_{sign}}, {\sf pk_{sanit}}_i)$ to a signing oracle of the underlying secure MQ-based signature oracle and receives the signature $\sigma_{i,1}$. In the next step, it sends $({\sf msg}_{i}, {\sf pk}_{{\sf sanit},i}, {\sf pk}_{{\sf sign},i})$ to a signing oracle of the underlying MQ-based signature to get back the signature ${\sigma_{i,2}}$ on the full message. It returns $\sigma=(\sigma_{i,1}, \sigma_{i,2}, {\sf AD}_i)$.
	\end{description}
	
	Finally $\mathcal{ B}$ parses $\sigma^*=(\sigma_1^*,\sigma_2^*, {\sf AD}^*)$, ${\sf msg^*_{fixed}}={\sf FIXED_{AD^*}(msg^*)}$ and returns the tuple $({\sf msg^*_{fixed}}, {\sf AD^*, pk_{sign}, pk^*_{sanit}, \sigma_1^*})$.
	
	\begin{description}
		\item Analyze: We show the if {$\mathcal{ G}$} wins its experiment, then $\mathcal{ B}$ also wins its experiment. If $\mathcal{ G}$ wins its experiment it means $\mathcal{G}$ successfully outputs $({\sf msg^*, \sigma^*, pk^*_{sanit}})$ with  
		
		\begin{align}
		&{\sf Verification({\sf msg^{*},\sigma^{*}, pk_{sign},pk^*_{sanit}})=true}\label{imut1}\\
		&\mbox{and } \forall i=1,2, \dots q,\;\;	 {\sf pk^*_{sanit}\neq pk_{{sanit},i}}\label{imut2}\\
		&\mbox{or } {\sf msg^{*}\notin \{MODIFY(msg_{i})\; |\; MODIFY \mbox{ with }}\nonumber \\ &{\sf AD}_i({\sf MODIFY})=1\}\label{imut3}
		\end{align}
		By equation \ref{imut1}, $\sigma^*$ is a valid signature, generated during the successful experiment ${\sf EXP\mbox{-}Immutability_{\mathcal{G}}^{SSS}}$ of adversary $\mathcal{G}$. Note that $\sigma^*=(\sigma_1^*, \sigma_2^*,{\sf AD}^*)$, and hence $\sigma^*$ contains a valid signature $\sigma_1^*$ on $({\sf 0, {msg}^*_{fixed}, AD^*,pk^*_{sanit} })$ which verifies correctly under the public key of the signer.
		
		To show that unforgeability of the underlying secure MQ-based signature is broken, it is enough to demonstrate that this tuple $({\sf 0, {msg}^*_{fixed}, AD^*,pk^*_{sanit} })$ has not been queried before by $\mathcal{B}$ through the signing oracle of the underlying secure MQ-based signature scheme. Assume on the contrary that if $({\sf 0, {msg}^*_{fixed}, AD^*,pk^*_{sanit} })=( 0,  {\sf msg}^*_{{\sf fixed}, i}, {\sf AD}^*_i,{\sf pk}^*_{{\sf sanit},i})$ for some $i$th query then ${\sf AD}_i={\sf AD^*}$ and ${\sf Fixed_{AD}(msg^*)}={\sf Fixed_{AD}}({\sf msg}_i)$. Thus, ${\sf  msg^*}$ must be a valid modification ${\sf MODIFY}({\sf msg}_i)$ of ${\sf msg}_i$, but it contradicts equation \ref{imut3}.
		
		Therefore,  if $\mathcal{G}$ wins the experiment with success probability $\delta(\kappa)=
		{\sf Prob}[{\sf EXP\mbox{-}Immutability_{\mathcal{G}}^{SSS}}=1]$, then $\mathcal{ B}$ also wins its experiment with non-negligible success probability greater than or equal to $\delta(\kappa)$
\end{description}\end{proof}

\begin{theorem}
	The proposed {\sf Mul-SAN} provides privacy.
\end{theorem}
\begin{proof}We will prove the privacy of our proposed design by using the indistinguishability notion discussed in \cite{brzuska2009security}. An adversary $\mathcal{ G}$ works as follows. $\mathcal{ G}$ chooses two pairs ${\sf (msg_0, {\sf MODIFY}_0), (msg_1, {\sf MODIFY}_1)}$ comprising the message ${\sf msg_0}$ and ${\sf msg_1}$ and their respective modifications together with the algorithm ${\sf AD}$. $\mathcal{G}$ is allowed to make queries to {\it Left or Right Oracle} $\mathcal{O}_{\sf LoR}$. $\mathcal{O}_{\sf LoR}$ oracle chooses a bit once, and in the following either output a sanitized signature for the left tuple ${\sf (msg_0,} {\sf MODIFY}_0)$ $(b=0)$ or sanitized signature for the right tuple $({\sf msg_1}, {\sf MODIFY}_1)$ $(b=1)$. The adversary $\mathcal{G}$ now predicts the random bit $b$. Note that for each query to $\mathcal{O}_{\sf LoR}$ if the resulting modified message is not identical for both left and right tuple and the modification made does not match ${\sf AD}$, the task of guessing the bit is straightforward. Therefore, we put an additional condition that modified messages {\sf MODIFY}$_0$, and {\sf MODIFY}$_1$ are identical for both of the tuples. We can now easily see that the output distribution of $\mathcal{O}_{\sf LoR}$ is identical for both values of the bit $b$. This is due to two reasons: firstly, the constraint we put earlier ensures that both tuples comprising of the message and the modification result in the same outcome, and secondly, the sanitiser does the signature on the modified message from scratch. Thus, we conclude that {\sf Mul-SAN} satisfies privacy.
\end{proof}

\begin{theorem}\label{thm:sanacc}
	The proposed sanitizable signature scheme {\sf Mul-SAN} is sanitizer-accountable, presuming that the underlying MQ-based signature scheme satisfies unforgeability.
\end{theorem}
\begin{proof}
	We assume, on the contrary, that our proposed design {\sf Mul-SAN} is not signer-accountable. Suppose that there exists an polynomial time adversary $\mathcal{G}_{\sf sanit}$ such that the advantage $\delta(\kappa)=
	{\sf Prob}[{\sf EXP\mbox{-}San\mbox{-}Acc_{\mathcal{G}_{sanit}}^{SSS}}]$ is non-negligible. We show how to build an adversary $\mathcal{B}$, which breaks the unforgeability of the underlying secure MQ-based signature scheme. Note that $\mathcal{B}$ has oracle access to a signing oracle of the underlying secure multivariate-based signature.
	
	\begin{description}
		\item {$\mathcal{ B}'s$ construction}: $\mathcal{ B}$ receives the public key ${\sf pk_{sign}}$ and runs  $ ({\sf pk^*_{sanit}, {\sf msg^*}, \sigma^*})\leftarrow \mathcal{G}_{\sf sanit}({\sf pk_{sign}})$. During the experiment, ${\cal B}$ simulates the  oracle ${\sf Signature(\cdot, sk_{sign}, \cdot, \cdot)}$ to $\mathcal{G}_{\sf sanit }$ as follows:
		\item {\sf Signature$(\cdot, {\sf sk_{sign}}, \cdot, \cdot)$}: On the $i$th input $({\sf msg}_{i}, {\sf AD}_i, {\sf pk}_{{\sf sanit},i})$, $\mathcal{ B}$ first computes the fixed message part ${\sf msg_{fixed}} \leftarrow {\sf  FIXED}_{{\sf AD}_i}({\sf msg}_i)$. It then sends $({\sf msg}_i, {\sf AD}_i, {\sf pk_{sign}}, {\sf pk_{sanit}}_i)$ to a signing oracle of the underlying secure MQ-based signature oracle and receives the signature $\sigma_{i,1}$. In the next step, it sends $({\sf msg}_{i}, {\sf pk}_{{\sf sanit},i}, {\sf pk}_{{\sf sign},i})$ to a signing oracle of the underlying MQ-based signature to get back the signature ${\sigma_{i,2}}$ on the full message. It returns $\sigma=(\sigma_{i,1}, \sigma_{i,2}, {\sf AD}_i)$.
	\end{description}
	
	Finally $\mathcal{ B}$ parses $\sigma^*=(\sigma_1^*,\sigma_2^*, {\sf AD}^*)$, ${\sf msg^*_{fixed}}={\sf FIXED_{AD^*}(msg^*)}$ and returns the tuple $({\sf msg^*_{fixed}}, {\sf AD^*, pk_{sign}, pk^*_{sanit}, \sigma_1^*})$.
	
	\begin{description}
		\item Analyze: We show the if {$\mathcal{ G}_{\sf sanit}$} wins its experiment, then $\mathcal{ B}$ also wins its experiment. If $\mathcal{ G}_{\sf sanit }$ wins its experiment it means $\mathcal{G}_{\sf sanit}$ successfully outputs $({\sf msg^*, \sigma^*, pk^*_{sanit}})$ with

		\begin{align}
		({\sf pk_{sanit}^*}, {\sf msg^*})\neq ({\sf pk}_{{\sf sanit},i}, {\sf msg'}_j) \forall i=1,2,\dots, q \label{sanacc1}\\
		{\sf Verification({\sf msg^{*},\sigma^{*}, pk_{sign},pk*_{sanit}})=true}\label{sanacc2}\\
		\mbox{and } {\sf Judge({\sf msg^{*},\sigma^{*}, pk_{sign},pk*_{sanit}})=Sig}\label{sanacc3}
		\end{align}
		By equation \ref{sanacc3}, we conclude that {\sf Judge} inspects $\sigma^*_2$ and checks that $\sigma^*_2$ is a correct signer's signature on $( {\sf 1, {msg}^*, AD^*,pk^*_{sanit}, pk_{sign} })$. But since by equation \ref{sanacc1}, ${\sf pk_{sanit}^*}, {\sf msg^*})\neq ({\sf pk}_{{\sf sanit},i}, {\sf msg'}_j) $ for all $i$ and as all signatures for the fixed part are signatures over the messages prepended with a 0-bit, it follows that $( {\sf 1, {msg}^*, AD^*,pk^*_{sanit}, pk_{sign} })$ has not been signed before.

		Therefore if $\mathcal{G}_{\sf sanit}$ wins the experiment with success probability $\delta(\kappa)=
		{\sf Prob}[{\sf EXP\mbox{-}Sanitizer\mbox{-}Acc{\mathcal{G}}^{SSS}}=1]$, then $\mathcal{ B}$ also wins its experiment with non-negligible success probability greater than or equal to $\delta(\kappa)$
\end{description}\end{proof}

We will now show that our proposed design is signed and accountable. The proof is very similar to Theorem \ref{thm:sanacc}, so we present a sketch of the proof.
\begin{theorem}\label{thm:signacc}
	The proposed sanitizable signature scheme {\sf Mul-SAN} satisfies the security property of signer-accountability if the underlying MQ-based signature scheme  employed in the sanitization algorithm is unforgeable.
\end{theorem}
\begin{proof}
	Let $\mathcal{G}_{\sf signer}$ be an adversary having access to the oracle {\sf Sanitization} with non-negligible success probability in the Game \ref{signeracc}. Let $({\sf pk^*_{sign}, {\sf msg^*}, \sigma^*})$ be a valid forgery against the signer accountability. It implies that (${\sf pk^*_{sign}}, {\sf msg^*})\neq ({\sf pk}_{{\sf sign},i}, {\sf msg'}_i)$ for all queries $({\sf msg'}_{i}, {\sf \sigma'}_i, {\sf pk}_{{\sf sanit},i})$ to the oracle {\sf Sanitization}. In addition, the three-tuple $({\sf pk^*_{sign}, {\sf msg^*}, \sigma^*})$ is such that {\sf Judge} looks at $\sigma^*_2$ and checks the validity of $\sigma^*_2$ for the message ($1, {\sf msg^*,\sigma^*, pk^*_{sanit},pk_{sign}}$). Note that a sanitizer only signs the message prepended with the bit $1$. In addition, we can observe that (${\sf pk^*_{sign}}, {\sf msg^*})\neq ({\sf pk}_{{\sf sign},i}, {\sf msg'}_i)$ for all the queries. Therefore, we can conclude that ${\sf msg^*}$ was not inputted by the sanitizer in the {\sf Mul-SAN.Sanitization} before. This implies the underlying multivariate-based signature scheme employed in the {\sf Mul-SAN.Sanitization} algorithm is not unforgeable.
\end{proof}
\begin{theorem}
	The proposed sanitizable signature scheme {\sf Mul-SAN} provides unforgeability under the assumption that the underlying MQ-based signature scheme is unforgeable.
\end{theorem}

\begin{proof}
	In \cite{brzuska2009security}, authors showed that unforgeability is implied if a SSS satisfies both signer and sanitizer accountability. Thus, our proposed design is unforgeable by Theorem \ref{thm:signacc} and \ref{thm:sanacc}.
\end{proof}

\section{Proposed Blockchain-Envisioned Solution For Audit Log Mechanism}

We now explore the possible application of {\sf Mul-SAN} in the Blockchain-enabled audit log mechanism. 
An enterprise can sign audit logs at the moment of generation and record this first signature on a Blockchain. A trusted person (say sanitizer) inside the organization can play the role of the sanitizer. Later, if the enterprise wishes to redact/modify a subset of segments in the audit logs, the sanitizer may utilize the {\sf Mul-SAN.Sanitization} to update the initial signature on the Blockchain. The Blockchain maintains an immutable log of the segments that were modified. When a receiver receives the logs, they can confirm the integrity of the received material at any given moment (after content creation) using the underlying verification algorithm of {\sf Mul-SAN}. In other words, they can determine which segments have been changed and ensure the sanitizer hasn't altered the remaining segments. Within the Blockchain network, receiver and enterprise-specific nodes are designated. A receiver submits a challenge to the company after needing to verify the integrity of the data. Next, in response to the received challenge, the company creates an auditing proof, which it then broadcasts to the Blockchain network. The representative nodes then bundle the auditing evidence into a new block, which they then broadcast by interacting with other nodes. Finally, the receiver can obtain the auditing proof from Blockchain and verify the auditing proof.

Blockchain offers several key advantages in integrity auditing:

\begin{itemize}
\item Decentralization: Blockchain operates as a distributed ledger maintained by miner nodes rather than relying on a central authority. This allows the assignment of integrity auditing tasks to various miners, reducing reliance on a single Trustworthy Third-Party Auditor (TPA) and minimizing the risk of a single point of failure.

\item Transparency and Tamer-Resistance: Records within the Blockchain are publicly accessible, enabling all miners and authorized parties to view and authenticate auditing outcomes. Each block in the Blockchain contains the hash value of its preceding block; therefore, any tampering activity can be prevented.

\item Non-repudiation: Blockchain maintains comprehensive operation logs and auditing histories, enabling the receiver to scrutinize these logs for potentially malicious behaviour.

\end{itemize}

\section{Efficiency Analysis\label{efficiency}}

We now present the storage and communication complexity of our proposed ${\sf Mul\mbox{-}SAN}$. The size of the public key of both the signer and sanitiser is $\frac{m(n+2)(n+1)}{2}$ field ($\mathbb{F}_q$) elements. $n^2+m^2+C$ field ($\mathbb{F}_q$) elements give the private key size. Here, $m$ denotes the number of equations, $n$ denotes the number of variables, and $C$ denote the size of the central map of the underlying secure MQ-based signature. We now have a look at the signature size. The signature size is $2n$ field ($\mathbb{F}_q$) elements. Table \ref{table:1} summarises the communication and storage complexity of {\sf Mul-SAN}.

\begin{table}[h!]
	\centering
	\caption{Summary of Communication and Storage Overheads of {\sf Mul-SAN}}

	\begin{tabular}{|l|l|}
		\hline
		Size of ${\sf pk_{sign}}$                                                         & $\frac{m(n+2)(n+1)}{2}$ field $ (\mathbb{F}_q)$ elements \\ \hline
		Size of  ${\sf sk_{sign}}$                                                        & $n^2+m^2+C$ field $(\mathbb{F}_q)$ elements  \\ \hline
		Size of  ${\sf pk_{sanit}}$                                                          & $\frac{m(n+2)(n+1)}{2}$ field $ (\mathbb{F}_q)$ elements  \\ \hline
		Size of  ${\sf sk_{sign}}$                                                  &  $n^2+m^2+C$ field $(\mathbb{F}_q)$ elements  \\ \hline
		Signature size                                                        & $2n$ field ($\mathbb{F}_q$) elements   \\ \hline
		Size of the sanitized signature & $2n$ field ($\mathbb{F}_q$) elements\\ \hline
		
	\end{tabular}
	
	\label{table:1}
\end{table}

To the extent of our knowledge, our proposed design {\sf Mul-SAN} is the \textit{first} post-quantum secure sanitizable signature. {\sf Mul-SAN} is also the {\it first} multivariate-based sanitizable signature in the current state of the art. Without other multivariate-based SSS, we compare our proposed design with other non-multivariate-based SSS. We instantiate {\sf Mul-SAN} with UOV signature scheme with parameters $(n=160,m = 64)$ \cite{beullens2023oil}. The chosen parameter set offers $128$-bit security over the field $GF(16)$. We refer to Table \ref{table:2} for a comparative summary.

\begin{table}[h!]
	\centering
	\caption{Comparison for $128$-bit security level over the field $GF(16)$}
	\adjustbox{scale=0.70}{%
		\begin{tabular}{|p{4cm}|p{1cm}|p{1cm}|p{1cm}|p{1cm}|p{1cm}|}
			\hline
			Scheme                                                                                      & Size of ${\sf sk_{sign}}$ (in Kb)  & Size of ${\sf sk_{sanit}}$ (in Kb)  &    Signer's signature size (in Kb) & Sanitizer's signature size (in Kb) & Post-quantum security       \\ \hline
			
			\begin{tabular}[c]{@{}c@{}} Bossuat2021 \cite{bossuat2021unlinkable}\\ \end{tabular}        &       0.12           & 0.12  & 1.89& 1.89 & $\times$\\ \hline
			\begin{tabular}[c]{@{}c@{}} Brzuska et al. \cite{brzuska2010unlinkability}\\ \end{tabular}         &       0.06           & 0.06  & 87.40 & 87.40 & $\times$\\ \hline
			\begin{tabular}[c]{@{}c@{}} Bultel2019 \cite{bultel2019efficient}\end{tabular}                                               & 0.12                  & 0.12   &1.35 &1.35& $\times$ \\ \hline
			\begin{tabular}[c]{@{}c@{}} Lai et al. \cite{lai2016efficient}\\ \end{tabular}                                             &                  0.32           & 0.10  & 55.71& 55.71& $\times$\\ \hline
			\begin{tabular}[c]{@{}c@{}}{\sf Mul-SAN } \\ \end{tabular}          &       2789.632         & 2789.632 & 0.77 &   0.77 &  $\checkmark$ \\ \hline

		\end{tabular}
	}
	\label{table:2}
\end{table}

We now discuss the findings of our comparative analysis. We can observe that our scheme performs better in signature size of both the signer and sanitizer when compared to the other schemes \cite{bossuat2021unlinkable,bultel2019efficient,lai2016efficient,brzuska2010unlinkability}. It is a well-known disadvantage with the MPKC-based scheme \cite{ding2020multivariate} that the size of the keys is enormous compared to classical cryptosystems like RSA or ECC. However, MPKC-based schemes are swift and can be efficiently employed on memory-constraint devices \cite{shim2020high,ferozpuri2018high,yi2018high,yi2018under,yi2019rainbow,nakkar2017fault,nakkar2017secure}. In addition, they are straightforward to implement. The only arithmetic operations required in MPKC are multiplications and additions over a finite field.

The second point to note is that unlike other non-multivariate based SSS \cite{bossuat2021unlinkable,brzuska2010unlinkability,bultel2019efficient,lai2016efficient}, our design is post-quantum secure. {\sf Mul-SAN} offers security against attacks by quantum computers. Schemes in \cite{bultel2019efficient,brzuska2010unlinkability,bossuat2021unlinkable,lai2016efficient} base their security on number-theoretic hardness assumptions. These schemes are insecure when an adversary gets access to large-scale quantum computers. Thus, schemes in \cite{bultel2019efficient,brzuska2010unlinkability,bossuat2021unlinkable,lai2016efficient} are not secure and will soon become obsolete. This gives our proposed design an edge over the existing SSS. 

\section{Conclusion}
Audit logs are significant in transparently tracking system events within corporate organizations and enterprise systems. However, scenarios often arise where these logs contain sensitive information or become excessively large. Dealing with the entirety of such data becomes impractical, making it more feasible to handle subsets of it. To address these challenges securely, we proposed a post-quantum secure multivariate-based SSS named {\sf Mul-SAN}. This design guarantees unforgeability, privacy, immutability, and signer and sanitizer accountability, contingent upon the assumption that solving the $MQ$ problem is NP-hard. Additionally, we explored the utilization of Blockchain in a tamper-proof audit log mechanism system.\\\\

\noindent \textbf{Declaration of competing interest}\\
The authors declare that they have no known competing financial interests or personal relationships that could have appeared to influence the work reported in this paper.\\
\noindent \textbf{Acknowledgment}\\
This work was supported by CEFIPRA CSRP project number 6701-1.\\
\noindent \textbf{Data availability}\\
Data sharing is not applicable to this article as no new data were generated or analyzed
to support this research.

\bibliographystyle{acm}
\bibliography{ref}

\end{document}